\documentclass[letterpaper,10pt,conference]{ieeeconf}
\IEEEoverridecommandlockouts
\usepackage[dvipsnames]{xcolor}
\usepackage{cite}
\usepackage{amsmath,amssymb,amsfonts}
\usepackage{algorithmic}
\usepackage{graphicx}
\usepackage[OT2,T1]{fontenc}
\usepackage{MnSymbol}
\usepackage{lipsum}
\usepackage{color}
\newtheorem{theorem}{Theorem}
\newtheorem{assumption}{Assumption}
\newtheorem{definition}{Definition}

\newtheorem{proposition}{Proposition}
\newtheorem{problem}{Problem}

\newtheorem{corollary}{Corollary}
\newtheorem{remark}{Remark}
\DeclareSymbolFont{cyrletters}{OT2}{wncyr}{m}{n}
\DeclareMathSymbol{\Sha}{\mathalpha}{cyrletters}{"58}
\usepackage{textcomp}
\begin{document}

\title{A harmonic framework for the identification of linear time-periodic systems}
\author{Flora Vernerey, Pierre Riedinger, Andrea Iannelli and Jamal Daafouz\\
\thanks{F. Vernerey, P. Riedinger and J. Daafouz are with Universit\'e de Lorraine, CNRS, CRAN, F-54000 Nancy, France. A. Iannelli is with University of Stuttgart, Institute for Systems Theory and Automatic Control.}}
\maketitle
%%%%%%%%%%%%%%%%%%%%%%%%%%%%%%%%%%%%%%%%%%%%%%%%%%%%%%%%%%%%%%%%%%%%%%%%%%%%%%%%
\begin{abstract}
This paper presents a novel approach for the identification of linear time-periodic (LTP) systems in continuous time. This method is based on harmonic modeling and consists in converting any LTP system into an equivalent LTI system with infinite dimension. Leveraging specific harmonic properties, we demonstrate that solving this infinite-dimensional identification problem can be reduced to solving a finite-dimensional linear least-squares problem. The result is an approximation of the original solution with an arbitrarily small error. Our approach offers several significant advantages. The first one is closely tied to the inherent LTI characteristic of the harmonic system, along with the Toeplitz structure exhibited by its elements. The second advantage is related to the regularization property achieved through the integral action when computing the phasors from input and state trajectories. Finally, our method avoids the computation of signals' derivative. This sets our approach apart from existing methods that rely on such computations, which can be a notable drawback, especially in continuous-time settings. We provide numerical simulations that convincingly demonstrate the effectiveness of the proposed method, even in scenarios where signals are corrupted by noise.\end{abstract}
%\begin{IEEEkeywords}
%, Dynamic phasors, Harmonic modeling and control
%\end{IEEEkeywords}

%%%%%%%%%%%%%%%%%%%%%%%%%%%%%%%%%%%%%%%%%%%%%%%%%%%%%%%%%%%%%%%%%%%%%%%%%%%%%%%%
\section{Introduction}

Periodicity arises naturally in various engineering and scientific disciplines \cite{farkas_periodic_1994}. From the mechanical vibrations of a car engine to the oscillations in electronic circuits, and even the rhythmic patterns of biological processes, periodicity is a ubiquitous feature \cite{allen_frequency-domain_2009, salis_stability_2017, bittanti_periodic_2009}. Understanding, analyzing, and controlling these periodic behaviors are essential for optimizing system performance, ensuring stability, and enhancing the predictability of various applications. 

Linear time-periodic (LTP) systems, a subset of linear time-varying systems, are characterized by periodic variations in their parameters over time. Furthermore, under specific conditions, nonlinear systems can be approximated as LTP systems when linearized along a periodic trajectory \cite{salis_stability_2017}. Given the broad spectrum of applications, modeling LTP systems is of significant interest for both analysis and control design. LTP systems, in comparison to LTI systems, introduce a higher degree of complexity. This complexity accounts for the preference in focusing on the identification of LTI systems, a choice supported by the extensive development of tools and methodologies for this purpose \cite{ljung1998system}. 

Most identification techniques for LTP systems involve first identifying one or more linear time-invariant (LTI) systems, which then serve as the basis for deriving a model of the LTP system. Lifting schemes have proven successful to identify the parameters of a discrete time equivalent LTI system either in the time domain \cite{markovsky_realization_2014} or the frequency domain \cite{saetti_identification_2019}. 
The subspace identification method has also been extended to LTP systems both in the time-domain \cite{Verhaegen1995} and more recently in frequency-domain \cite{Uyanik2019,yin_subspace_2021}.
%The subspace identification method has been extended to LTP systems with periodic inputs \cite{yin_subspace_2021, cadoret_linear_2022}.
 In \cite{yin_subspace_2021}, restrictions on the input-output dimension and input form were relaxed by leveraging the idea to employ the frequency response of a time-lifted system with a linear time-invariant structure.
%To overcome the issues related to the
%computation of frequency response of LTP systems and estimate the time-aliased periodic impulse response, the idea is to employ the frequency response of a time-lifted system with a linear time-invariant structure. 
Also, in \cite{uyanik_parametric_2016}, Fourier transformations of state and input data are used to identify the Fourier series coefficients of the state and input matrices. It is important to note that this approach is restricted to stable systems with oscillations that attain a steady-state, which can be a limiting factor, particularly in control applications. 

Moreover, the majority of the proposed methods to date is available for discrete-time systems and, when extended to continuous-time, require unquantified approximations. Our aim is to propose a methodology to identify the Fourier coefficients of the state and input matrices of LTP systems in continuous time for both stable and unstable systems. Importantly, we aim to remove restrictions related to systems achieving a steady-state and to eliminate signal limitations.
To achieve this objective, we make use of an equivalence result established in \cite{blin_necessary_2022}, linking LTP systems to infinite-dimensional LTI systems characterized by a block Toeplitz structure formed by the Fourier coefficients of the state and input matrices. LTP system's parameters are thus inferred from the Fourier series. However, it is essential to note that the harmonic system is inherently infinite-dimensional. Consequently, truncation becomes necessary. By selecting a sufficiently high truncation order, the identification problem is translated into a finite-dimensional linear least-squares problem. We prove that the solution to this finite-dimensional problem converges to the solution of the infinite-dimensional counterpart with an arbitrarily small error.

The paper is organized as follows. The next section is dedicated to mathematical preliminaries on harmonic modelling. In Section III, we state the identification problem. The main results are established in Section IV where we tackle the approximation of the infinite-dimensional identification problem. Illustrative examples are given in Section V, demonstrating the application of our approach to identify linear time-periodic systems, even in scenarios where state measurements are affected by noise.

{\bf Notations: } %The transpose of a matrix $A$ is denoted $A'$ and $A^*$ denotes the complex conjugate transpose $A^*=\bar A'$. 
$C^a$ denotes the space of absolutely continuous function,
$L^{p}([a\ b],\mathbb{C}^n)$ (resp. $\ell^p(\mathbb{C}^n)$) denotes the Lebesgues spaces of $p-$integrable functions on $[a, b]$ with values in $\mathbb{C}^n$ (resp. $p-$summable sequences of $\mathbb{C}^n$) for $1\leq p\leq\infty$. $L_{loc}^{p}$ is the set of locally $p-$integrable functions. The notation $f(t)=g(t)\ a.e.$ means almost everywhere in $t$ or for almost every $t$. 
To simplify the notations, $L^p([a,b])$ or $L^p$ will be often used instead of $L^p([a,b],\mathbb{C}^n)$. 
%For example, $x\in L^2([a,b])$ means $x \in L^2([a,b],\mathbb{C}^n)$. %We denote by $col(X)$ the vectorization of a matrix $X$, formed by stacking the columns of $X$ into a single column vector. Finally, $<\cdot,\cdot>$ refers to the scalar product in $\ell^2$.
%%%%%%%%%%%%%%%%%%%%%%%%%%%%%%%%%%%%%%%%%%%%%%%%%%%%%%%%%%%%%%%%%%%%%%%%%%%%%%%%
\vspace{-.3cm}
\section{Preliminaries on Harmonic modelling}
%Metric units are preferred for use in IEEE publications in light of their
%international readership and the inherent convenience of these units in many fields.
%In particular, the use of the International System of Units (SI Units) is advocated.
% This system includes a subsystem the MKSA units, which are based on the
% meter, kilogram, second, and ampere. British units may be used as secondary units
% (in parenthesis). An exception is when British units are used as identifiers in trade,
% such as, 3.5 inch disk drive.
%We start by recalling some preliminaries related to Toeplitz block matrices, sliding Fourier decomposition in the context of harmonic modeling and the trace operator.

For a given integer $n$, consider $x\in L^{2}_{loc}(\mathbb{R},\mathbb{C}^n)$ a complex valued function of time. Its sliding Fourier decomposition over a window of length $T$ is defined by the time-varying infinite sequence $X=(\cdots,X_{-1},X_0,X_1,\cdots):=\mathcal{F}(x)\in C^a(\mathbb{R},\ell^2(\mathbb{C}^n))$ (see \cite{blin_necessary_2022}) whose $n$-dimensional components $X_k$ (named $k-$th phasor) satisfy:
$$X_{k}(t):=\frac{1}{T}\int_{t-T}^t x(\tau)e^{-\textsf{j}\omega k \tau}d\tau$$ for $k\in \mathbb{Z}$, with $\omega:=\frac{2\pi}{T}$.
 The Toeplitz transformation of a matrix function $A\in L^{2}_{loc}(\mathbb{R},\mathbb{C}^{n\times m})$, denoted $\mathcal{A}:=\mathcal{T}(A)$, defines a block Toeplitz and infinite dimensional matrix as follows: 
\begin{align}
	\mathcal{A}:=\mathcal{T}(A)=
	\left[
	\begin{array}{ccccc}
		\ddots & & \vdots & &\udots \\ & A_{0} & A_{-1} & A_{-2} & \\
		\cdots & A_{1} & A_{0} & A_{-1} & \cdots \\
		& A_{2} & A_{1} & A_{0} & \\
		\udots & & \vdots & & \ddots\end{array}\right],\label{top}\end{align}
where $(A_{k})_{k\in\mathbb{Z}}$ denotes the phasor sequence of $A$.\\
\begin{definition}\label{H} We say that $X$ belongs to $H$ if $X$ is an absolutely continuous function (i.e $X\in C^a(\mathbb{R},\ell^2(\mathbb{C}^n))$ and fulfills for any $k$ the following condition: \begin{equation}\dot X_k(t)=\dot X_0(t)e^{- \textsf{j}\omega k t} \ a.e. \label{ordre_k_ordre_0}\end{equation} 

\end{definition}

Similarly to the Riesz-Fisher theorem which establishes a one-to-one correspondence between the spaces $L^2$ and $\ell^2$, the following theorem establishes a one-to-one correspondence between the spaces $L_{loc}^2$ and $H$ (see \cite{blin_necessary_2022}).
\begin{theorem}\label{coincidence}For a given $X\in L_{loc}^{\infty}(\mathbb{R},\ell^2(\mathbb{C}^n))$, there exists a representative $x\in L^2_{loc}(\mathbb{R},\mathbb{C}^n)$ of $X$, i.e. $X=\mathcal{F}(x)$, if and only if $X \in H$.
\end{theorem}
Thanks to Theorem~\ref{coincidence}, it is established in \cite{blin_necessary_2022} that any system having solutions in Carath\'eodory sense can be transformed by a sliding Fourier decomposition into an infinite dimensional system for which a one-to-one correspondence between their respective trajectories is established providing that the trajectories in the infinite dimensional space belong to the subspace $H$. Moreover, when a $T-$periodic system is considered, the resulting infinite dimensional system is time-invariant. For instance, consider $T-$periodic functions $A(\cdot)$ and $B(\cdot)$ respectively of class $L^2([0\ T],\mathbb{C}^{n\times n})$ and $L^{\infty}([0\ T],\mathbb{C}^{n\times m})$ and let: 
\begin{align}\dot x(t)=A(t)x(t)+B(t)u(t)\quad x(0)=x_0\label{ltp}\end{align}
If $x$ is a solution associated to the control $u\in L_{loc}^2(\mathbb{R},{\mathbb{C}^m)}$ of the linear time periodic (LTP) system ~\eqref{ltp} then, $X:=\mathcal{F}(x)$ is a solution associated to $U:=\mathcal{F}(u)$ of the linear time invariant (LTI) system:
\begin{align}
	\dot X(t)=(\mathcal{A}-\mathcal{N})X(t)+\mathcal{B}U(t), \quad X(0)=\mathcal{F}(x)(0) \label{ltih}
\end{align}
where $\mathcal{A}=\mathcal{T}(A)$, $\mathcal{B}=\mathcal{T}(B)$ and 
\begin{equation}\mathcal{N}:=diag( \textsf{j}\omega k \otimes Id_n,\ k\in \mathbb{Z})\label{q}\end{equation}
with $\otimes$ the Kronecker product and $Id_n$ the identity matrix of size $n$.
Reciprocally, if $X\in H$ is a solution to \eqref{ltih} with $U\in H$, then their representatives $x$ and $u$
(i.e. $X=\mathcal{F}(x)$ and $U=\mathcal{F}(u)$) are a solution to~\eqref{ltp}. In addition, it is proved in \cite{blin_necessary_2022} that one can reconstruct time trajectories from the exact formula:
\begin{align}\label{recos} x(t)&=\mathcal{F}^{-1}(X)(t)=\sum_{k=-\infty}^{+\infty} X_k(t)e^{ \textsf{j}\omega k t}+\frac{T}{2}\dot X_0(t)\end{align}
%where $X_{k}=(X_{1,k}, \cdots, X_{n,k})$ for any $k\in \mathbb{Z}$.

%Note that if $M\in S^n_{+}$, $tr(\mathcal{M})=0$ implies that \textcolor{red}{$M(t)=0\ a.e$ and thus} $\mathcal{M}=0$. \textcolor{red}{So, it is straightforward to check that .}

\section{Problem formulation}

We consider the continuous-time LTP system
\eqref{ltp} where $A$ and $B$ are real-valued functions both in $L^\infty ([0 \ T])$, $x(t) \in \mathbb{R}^n$ and $u(t) \in \mathbb{R}^m$. Furthermore, we make the following assumptions: 
\begin{assumption}\label{assump1} 
The period $T$ is known and %and they do not change over time. \label{assump2}
the state can be measured or estimated over a time interval $[t_0,t_f]$. \label{assump3}
\end{assumption}

%\begin{assumption}
% The values of the truncation orders $d^oA$ and $d^oB$ are known or an upper bound on them can be estimated. \label{assump4}
%\end{assumption}

These assumptions are common for identification of LTP systems, see for example \cite{uyanik_parametric_2016}. 
%Our results could most certainly be extended to the case where the dimensions of the state and input vary over time but it does not fall within the scope of the current work. As for assumption \ref{assump4}, if the exact amount of non-vanishing harmonics is unknown, applying a high order truncation is a widespread strategy, even if it requires more computation.
%With these assumptions, notice that the values of the functions $A$ and $B$ can be calculated at any time instant $t$ as long as their phasors are known. Therefore, the behaviour of the LTP system \eqref{ltp} is entirely characterised by the phasors that we wish to determine, which shows the interest of solving problem \ref{prob}.
Building upon the preceding section, studying such LTP systems essentially involves examining an equivalent infinite-dimensional Linear Time-Invariant (LTI) system as defined in \eqref{ltih}. Indeed, as \begin{equation}
A(t)= \sum_{k = -\infty}^{\infty} A_{k} e^{\mathbf{j} \omega k t} \ a.e. \quad B(t)= \sum_{k = -\infty}^{\infty} B_{k} e^{\mathbf{j} \omega k t} a.e, \label{AB}
\end{equation} the identification of $A$ and $B$ can be achieved though the identification of the phasor sequences $(A_k)$ and $(B_k)$ that appear explicitly in the operators $\mathcal{A}$ and $\mathcal{B}$ in \eqref{ltih} (see \eqref{top}). Hence, it is possible to reformulate the task of identifying a LTP system as the task of identifying a LTI system, with the caveat that we must address the challenge of its infinite dimensionality. Consequently, the infinite-dimensional problem we aim to solve, achieving precision to an arbitrarily small error, can be stated as follows:

\begin{problem}\label{prob1}Under Assumption 1, 
identify $\mathcal{A}$ and $\mathcal{B}$ in \eqref{ltih} from a state/input $(x,u)$ trajectory of system \eqref{ltp} and deduce $A$ and $B$ with \eqref{AB}. 
\end{problem} 
Problem \ref{prob1} is an infinite dimensional one. Our goal is to provide a computationally tractable approach to identification of \eqref{ltih} with a guaranteed bound on the approximation. 
As we will see in the sequel, adopting this approach offers several noteworthy advantages. The first advantage is inherently tied to the LTI characteristic of the harmonic system, as well as the Toeplitz structure exhibited by its elements. The second advantage pertains to the filtering property achieved through the integral action for computing the phasors $(X,U)$ from state/input trajectories $(x,u)$. The final advantage stems from the fact that we do not require to compute the derivative terms. Indeed, these terms satisfy the following relationship:

$$\dot X_k(t)=e^{-j\omega k t}\dot X_0(t)$$ with \begin{equation}
 \dot X_0(t)=\frac{1}{T}(x(t)-x(t-T))\label{dotX0}
\end{equation} (see \cite{blin_necessary_2022}) and thus it is not necessary to compute any derivative of $x$. These two last properties will be very useful to generate harmonic data and simplify the identification problem.

\vspace{-.05cm}
\section{Main results}
 In this section, we show how problem \ref{prob1} can be reduced to a finite dimensional least squares identification problem whose solution satisfies the original problem up to an arbitrarily small error. As a consequence, the identification of the original LTP system is achieved with a guaranteed bound on the approximation. In addition, a direct bound on the difference of the estimated and true matrices is derived.
\subsection{A central strip identification problem}
Before deriving a finite dimensional approximation to Problem 1, we start by proving some technical results. 
\begin{theorem}\label{tech1} Let $A\in L^\infty([0\ T])$ and the Fourier series of $A$ given by $A(t)=\sum_{k = -\infty}^{\infty} A_{k} e^{\mathbf{j} \omega k t} \ a.e.$ Then, the operator sequence indexed by $p$, $ A|_p(t):=\sum_{k = -p}^{p} A_{k} e^{\mathbf{j} \omega k t}$ converges in $L^\infty$ operator norm to $A$ and we have:
$$\lim_{p\rightarrow +\infty}\|A-A|_{p}\|_{L^\infty}=\lim_{p\rightarrow +\infty}\|\mathcal{A}-\mathcal{A}|_{p}\|_{\ell^2}=0$$
where $\mathcal{A}:=\mathcal{T}(A)$ and $\mathcal{A}_p:=\mathcal{T}(A|_p)$.
\end{theorem}
\begin{proof}First, let us recall \cite{gohberg_classes_2013} that $A \in L^\infty ([0 \ T])$ if and only if $\mathcal{A}$ is a { bounded operator on $\ell^2$ i. e. there exists $C$ s.t. $$
\|\mathcal{A}\|_{\ell^2}=\sup_{\|x\|_{\ell^2}=1}\|\mathcal{A}x\|_{\ell^2}\leq C$$
 and $\|A\|_{L^\infty}=\|\mathcal{A}\|_{\ell^2}$.}
Hence, using the Fourier series of $A$,
%: $$A(t)=\sum_{k\in\mathbb{Z}} A_ke^{ \textsf{j}\omega kt} \ a.e.,$$ 
we can write: 
\begin{align}
\|\mathcal{A}-\mathcal{A}|_{p}\|_{\ell^2}& =\|A-A|_p\|_{L^\infty}=\|\sum_{|k|>p} A_ke^{ \textsf{j}\omega kt} \|_{L^\infty}\label{e1}
\end{align} where $\mathcal{A}:=\mathcal{T}(A)$ and $\mathcal{A}_p:=\mathcal{T}(A|_p)$. 
As by assumption there exists a constant $C_1$ such that
\begin{align*}
\|\mathcal{A}\|_{\ell^2}=\|A\|_{L^\infty}
&= \|\sum_{k\in \mathbb{Z}} A_ke^{ \textsf{j}\omega kt} \|_{L^\infty}<C_1
\end{align*}
the series $\sum_{k\in \mathbb{Z}} A_ke^{ \textsf{j}\omega kt}$ converges almost everywhere and $\lim_{p\rightarrow +\infty}\sum_{|k|>p} A_ke^{ \textsf{j}\omega kt}=0\ a.e.$ 
Taking the limit w.r.t. $p$ in \eqref{e1} leads to the result.
\end{proof}
From this result, it follows that replacing the pair $(\mathcal{A},\mathcal{B})$ by $(\mathcal{A}|_p,\mathcal{B}|_p)$ leads to the following approximation of the harmonic state dynamic: 
\begin{corollary}\label{conver2}
 Let $(x,u)$ be a trajectory of \eqref{ltp} and $(X,U):=\mathcal{F}(x,u)$. 
 For any $\epsilon>0$, there exists $p$ such that for any compact time interval $I$ \begin{equation}
  \sup_{t\in I}\frac{\|\Psi(X(t),U(t))\|_{\ell^2}}{\|(X(t),U(t))\|_{\ell^2}}\leq 2\epsilon \label{eps0}
 \end{equation}
 where $\Psi(X(t),U(t)):=(\mathcal{A}-\mathcal{N})X(t)+\mathcal{B}U(t)-((\mathcal{A}|_p-\mathcal{N})X(t)+\mathcal{B}|_pU(t))$
 %$M=\sup_{t\in I}(\|X(t)\|_{\ell^2}+\|U(t)\|_{\ell^2})$
\end{corollary}
\begin{proof}
 As the operator norms are sub-multiplicative, we have for any $t$:
  \begin{align*}
   \|(\mathcal{A}-\mathcal{N}) X(t)+\mathcal{B} U(t)-((\mathcal{A}|_p-\mathcal{N}) X(t)+\mathcal{B}|_p U(t))\|_{\ell^2}\\\leq
   \|\mathcal{A}-\mathcal{A}|_p\|_{\ell^2}\|X(t)\|_{\ell^2}
   +\|\mathcal{B}-\mathcal{B}|_p\|_{\ell^2} \| U(t)\|_{\ell^2}\\
   \leq
   (\|\mathcal{A}-\mathcal{A}|_p\|_{\ell^2}\|+\|\mathcal{B}-\mathcal{B}|_p\|_{\ell^2})) \| (X(t),U(t))\|_{\ell^2}
  \end{align*}
  then, as $X$ and $U$ are absolutely continuous functions of the time, the supremum on any compact set exists and using Theorem~\ref{tech1} the result follows. 
\end{proof}

\begin{definition}
 The degree $d^oA\in \mathbb{Z}^+\cup\{+\infty\}$ of $A$ is defined by the greatest non vanishing phasor of $A$.
\end{definition}
Corollary~\ref{conver2} shows that if $d^oA$ and $d^oB$ are not finite, it is always possible to obtain an accurate approximated solution to the identification problem by imposing sufficiently large $d^oA$ and $d^oB$.
Indeed for a given $\epsilon>0$ and $p(\epsilon)$ such that Corollary~\ref{conver2} holds, the solution $(\tilde{\mathcal{A}},\tilde{\mathcal{B}})$ of the normalized linear least- squares optimization problem:
$$\min_{\mathcal{A}|_p,\mathcal{B}|_p}\sup_{t\in I}\frac{\|\dot X(t)-((\mathcal{A}|_p-\mathcal{N})X(t)+\mathcal{B}|_pU(t))\|_{\ell^2}^2}{\|(X(t),U(t))\|_{\ell^2}^2}$$ where the unknowns $\mathcal{A}|_p$ $\mathcal{B}|_p$ which are block $p$-banded (i.e. $A_k:=0$ and $B_k:=0$ for $k>p$) and Toeplitz matrices, will necessarily satisfy relation~\eqref{eps0}.\\
Now, let us show that Problem \ref{prob1} can be simplified and reduced to a finite dimensional problem. This problem is referred to as the "central strip identification problem," relating to the "0-row" of (\ref{ltih}).
\begin{theorem}\label{centrale}Problem \ref{prob1} can be reduced to the central strip identification problem which involves the identification of:
\begin{equation}
  \dot X_0(t)=\sum_{k\in\mathbb{Z}}A_kX_{-k}(t)+\sum_{k\in\mathbb{Z}}B_kU_{-k}(t)\label{strip0}
\end{equation}
Moreover, if $d^oA$ and $d^oB$ are finite then the central strip identification problem reduces to the identification of:
$$\dot X_0(t)=\sum_{k=-d^oA}^{d^oA}A_kX_{-k}(t)+\sum_{k=-d^oB}^{d^oB}B_kU_{-k}(t)$$
which is a finite dimensional problem that involves $n^2(2d^oA + 1) + nm(2d^oB + 1)$ real unknowns.
\end{theorem}
\begin{proof}
 As $\dot X_p(t)=e^{-j\omega p t}\dot X_0(t)$ for any $p$ and as $$\dot X_0(t)=\sum_{k\in\mathbb{Z}}A_kX_{-k}(t)+\sum_{k\in\mathbb{Z}}B_kU_{-k}(t),$$ it follows
that the $p-$strip of \eqref{ltih} corresponding to $\dot X_p(t)$ satisfies
$$\dot X_p(t)=e^{-j\omega p t}(\sum_{k\in\mathbb{Z}}A_kX_{-k}(t)+\sum_{k\in\mathbb{Z}}B_kU_{-k}(t))$$
and is no more informative than $\dot X_0$. Thus, only the central strip is sufficient to identify the harmonic system \eqref{ltih}. 
 Moreover, if $d^oA$ and $d^oB$ are finite, there are $n^2 (2d^oA+1) + nm (2d^oB+1)$ complex unknowns to be determined in this problem. As $A$ and $B$ are real-valued functions, their phasors of negative and positive order are complex conjugates of each other: $\forall k \in \mathbb{Z}, A_{-k} = \overline{A_k}$. Furthermore, as any phasor of order $0$ is real-valued, this means that Problem \ref{prob1} amounts to identifying $n^2 (2d^oA+1) + nm (2d^oB+1)$ real unknowns.
\end{proof}

Finally, combining Corollary~\ref{conver2} and Theorem~\ref{centrale} leads to the following result:
\begin{theorem}\label{conv3}
For any $u\in L_{loc}^2$ and any $\epsilon>0$, there exists $p$ such that the finite dimensional normalized central strip identification problem on interval $I$ given by: 
\begin{equation}
 \min_{A_k, B_k, |k|\leq p}\sup_{t\in I}\frac{\|\dot X_0(t)-\sum_{k=-p}^{p}(A_kX_{-k}(t)+B_kU_{-k}(t))\|^2}{\|(X(t),U(t))\|_{\ell^2}^2}\label{opt}
\end{equation} where $\|\cdot\|$ refers to the 2-norm, leads to an approximated solution $\tilde A(t):=\sum_{k = -p}^{p} \tilde A_{k} e^{\mathbf{j} \omega k t}$ and $\tilde B(t):=\sum_{k = -p}^{p} \tilde B_{k} e^{\mathbf{j} \omega k t}$ that satisfies:
\begin{equation}
  \sup_{t\in I}\frac{\|\dot X_0(t)-\sum_{k=-p}^{p}(\tilde A_kX_{-k}(t)+\tilde B_kU_{-k}(t))\|}{\|(X(t),U(t))\|_{\ell^2}}\leq 2\epsilon \label{eps}
\end{equation}
and relation \eqref{eps} still holds true if $I$ in \eqref{opt} and \eqref{eps} is replaced by a discrete set $I_d\subset I$. 
\end{theorem}
\begin{proof}
 Let us define the central strip selecting operator $\mathcal{C}_0:=[\cdots\ 0 \ Id_n \ 0 \ \cdots]$ such that $X_0=\mathcal{C}X$ where $X:=\mathcal{F}(x)$.
 Obviously $\|\mathcal{C}_0\|=1$ where $\|\cdot\|$ is the standard Euclidian norm (2-norm) since 
 $$\sup_{\|X\|_{\ell^2}=1}\|\mathcal{C}_0X\|=\sup_{\|X\|_{\ell^2}=1}\|X_0\|=1$$ is achieved for $X=[\cdots, 0, X_0, 0 ,\cdots]$ with $\|X_0\|=1$.
 Now, for a given $u$, and $\epsilon>0$, we know that there exists $p$ such that Corollary~\ref{conver2} is satisfied and thus relation \eqref{eps0} holds.
 Therefore, \begin{align}
  \frac{\|\mathcal{C}_0\Psi(X(t),U(t))\|}{\|(X(t),U(t))\|_{\ell^2}}\leq 
 \|\mathcal{C}_0\|\sup_{t\in I}\frac{\|\Psi(X(t),U(t))\|_{\ell^2}}{\|(X(t),U(t))\|_{\ell^2}}
 \leq 2\epsilon. \label{eps2}
 \end{align}
 As $\mathcal{C}_0\Psi(X(t),U(t))=\dot X_0(t)-\sum_{k=-p}^{p} (A_kX_{-k}(t)+ B_kU_{-k}(t))$, we see that the minimizer of \eqref{opt} satisfies necessarily \eqref{eps}.
 Finally, if in \eqref{opt} $I$ is replaced by a discrete set $I_d\subset I$, the minimizer of \eqref{opt} satisfies necessarily \eqref{eps} on $I_d$. 
\end{proof}
\begin{remark}\label{normalization}
 Note that computing $\|(X(t),U(t))\|_{\ell^2}$ in \eqref{opt} can be simply achieved by computing $\|(x(t),u(t))\|_{L^2(([t-T \ t])}$ as Riesz-Fisher theorem implies $$\|X(t)\|_{\ell^2}=\|x\|_{L^2([t-T \ t])}=(\frac{1}{T}\int_{t-T}^tx^2(\tau)d\tau)^\frac{1}{2}.$$
\end{remark}

Having successfully converted the infinite-dimensional harmonic identification problem into an approximate finite-dimensional counterpart, the next subsection is dedicated to the conditions essential for achieving a precise solution from a discrete-time sequence of data.

\subsection{Solving the central strip least squares identification problem}
For a sufficiently large number $N$, consider a sampled state/input trajectory $(x,u)$ of the LTP system \eqref{ltp} with sampling time $\delta t=T/N$ over the time interval $I:=[t_0,t_f]$. 
For a given order $p$, the computation of phasors $(X_k(t),U_k(t))$ for $|k|\leq p$ on the time interval $[t_0 + T, t_f]$ can be performed using a Fast Fourier Transform. Also, $\dot X_0$ can be determined using \eqref{dotX0}. Then, these data are normalized following Remark~\ref{normalization} that is, for $|k|\leq p$
\begin{align*}
 ( X_k(t),U_k(t))_\mathbf{N}&:=\frac{( X_k(t),U_k(t))}{M(t)} \text{ and } \dot X_{0_\mathbf{N}}(t):=\frac{\dot X_0(t)}{M(t)}
\end{align*}
where $M(t):=\|(x(t),u(t))\|_{L^2([t-T \ t])}$. The data are stored as follows:
\begin{align*}
\mathbf{X_1} &:= \begin{pmatrix}
 \dot{X}_{{0}_\mathbf{N}}(t_0 + T) & \hdots & \dot{X}_{{0}_\mathbf{N}}(t_f)
\end{pmatrix} \in \mathbb{C}^{n \times L} \\
\mathbf{X_0} &:= \begin{pmatrix}
 X_{{-p:p}_\mathbf{N}}(t_0 +T) & \hdots & X_{{-p:p}_\mathbf{N}}(t_f)
\end{pmatrix} \in \mathbb{C}^{n(2p+1) \times L} \\
\mathbf{U_0} &:= \begin{pmatrix}
 U_{{-p:p}_\mathbf{N}}(t_0 +T) & \hdots & U_{{-p:p}_\mathbf{N}}(t_f)
\end{pmatrix} \in \mathbb{C}^{m(2p+1) \times L}
\end{align*}
where $X_{{-p:p}_\mathbf{N}}(t)$ (respectively $U_{{-p:p}_\mathbf{N}}(t)$) is a column vector which contains the normalized phasors from order $-p$ to $p$ of the $n$ components of $X$ (respectively $U$), $L$ being the number of samples. The $p$-banded central strip identification problem given by \eqref{opt} is formulated as follows:
\begin{problem}\label{lsqr}From data $(\mathbf{X_1},\mathbf{X_0} ,\mathbf{U_0})$, solve the least-squares problem:  
\begin{equation}
  \min_{A_k,B_k,\ |k|\leq p} \|\mathbf{X_1}-[A_p,\cdots,A_{-p}, B_p,\cdots,B_{-p}]\begin{bmatrix} \mathbf{X_0}\\\mathbf{U_0}
\end{bmatrix}\|^2 \label{lsqrpb}
 \end{equation}
where $\|\cdot\|$ refers to $2-$norm.  
\end{problem}
We have the following proposition. 
\begin{proposition} 
 The data $(\mathbf{X_{0}},\mathbf{U_{0}})$ are informative for system identification if and only if
 \begin{equation}
  rank \begin{pmatrix}
   \mathbf{X_{0}} \\ \mathbf{U_{0}}
  \end{pmatrix} = (n+m)(2p+1). \label{rank}
 \end{equation} 
\end{proposition}
\begin{proof} The proof follows from \cite{van_waarde_data_2020} dedicated to data informativity for noise free linear systems identification.
\end{proof}
To ensure that the rank condition (\ref{rank}) is met, a necessary condition is that the sample size, denoted as $L$, exceeds the value of $(n+m)(2p+1)$.
\begin{theorem} For a given $p$, if the data $(\mathbf{X_{0}},\mathbf{U_{0}})$ are informative for system identification, the solution of the least-squares problem is given by
$$[\tilde A_p,\cdots,\tilde A_{-p}, \tilde B_p,\cdots,\tilde B_{-p}]:=\mathbf{X_1}\begin{bmatrix} \mathbf{X_0}\\\mathbf{U_0}
\end{bmatrix}^\dag$$ where $\begin{bmatrix} \mathbf{X_0}\\\mathbf{U_0}
\end{bmatrix}^\dag$ refers to the pseudo-inverse of $\begin{bmatrix} \mathbf{X_0}\\\mathbf{U_0}
\end{bmatrix}$ and is uniquely determined. 
Moreover, for a given $\epsilon>0$, if $p$ satisfies Corollary~\ref{conver2} on $I:=[t_0+T,t_f]$
then relation \eqref{eps} is satisfied at least at every sample time $t:=t_k+T$ of $I$ and we have:
 \begin{align}
 \max_{t}&\|([\tilde A_p,\cdots,\tilde A_{-p}, \tilde B_p,\cdots,\tilde B_{-p}]\nonumber\\&-[ A_p,\cdots, A_{-p}, B_p,\cdots,B_{-p}])\begin{bmatrix}
   X_{{-p:p}_\mathbf{N}}(t) \\ U_{{-p:p}_\mathbf{N}}(t)
  \end{bmatrix}
\|\leq 4\epsilon. \label{error}
\end{align} %with $\|(X(t),U(t))_\mathbf{N}\|_{\ell^2}:=\frac{\|(X(t),U(t))\|_{\ell^2}}{M(t)}}$.

%A bound on \begin{align*}
 % \|([\tilde A_p,\cdots,\tilde A_{-p}, \tilde B_p,\cdots,\tilde B_{-p}]&-[ A_p,\cdots, A_{-p}, B_p,\cdots,B_{-p}])\|\\
%&\leq 4\epsilon \sum_{i=1}^{(n+m)(2p+1)} \alpha_i
%\end{align*}
%with $$\sum_{i=1}^{(n+m)(2p+1)} \alpha_i\begin{bmatrix}
%   X_{{-p:p}_\mathbf{N}}(t_i) \\ U_{{-p:p}_\mathbf{N}}(t_i)
%  \end{bmatrix}$$
%and thus there exists a finite number $M$ such that $$\|([\tilde A_p,\cdots,\tilde A_{-p}, \tilde B_p,\cdots,\tilde B_{-p}]-[ A_p,\cdots, A_{-p}, B_p,\cdots,B_{-p}])\|\leq 4 M\epsilon $$
\end{theorem}
\begin{proof}As condition \eqref{rank} is satisfied, the minimizer of \eqref{lsqrpb} is uniquely defined. If $p$ is such that Corollary~\ref{conver2} is satisfied on $I$ then relation \eqref{eps} is satisfied on every sample time $t:=t_k+T$ (as stated in Theorem~\ref{conv3}).
Using \eqref{eps} and \eqref{eps2}, it follows that for any sample time $t:=t_k+T$:
\begin{align*}
 \|([\tilde A_p,&\cdots,\tilde A_{-p}, \tilde B_p,\cdots,\tilde B_{-p}]\\
 &-[ A_p,\cdots, A_{-p}, B_p,\cdots,B_{-p}])\begin{bmatrix}
   X_{{-p:p}_\mathbf{N}}(t) \\ U_{{-p:p}_\mathbf{N}}(t)
  \end{bmatrix}\|\\
&\leq \|\sum_{k=-p}^{p}(\tilde A_kX_{-k}(t)+\tilde B_kU_{-k}(t))-\dot X_0(t)\|+\\
&\|\dot X_0(t)-\sum_{k=-p}^{p} (A_kX_{-k}(t)+B_kU_{-k}(t))\|\leq 4\epsilon.
\end{align*}
\end{proof}
\begin{corollary}
There is a subsequence of sampling time $t_{i_k}$ of length $(n+m)(2p+1)$ such that the matrix $V$ whose columns are formed by $(X_{{-p:p}_\mathbf{N}}(t_{i_k}),U_{{-p:p}_\mathbf{N}}(t_{i_k}))$ is invertible. Then, the following bound holds:
 \begin{align*}\|[\tilde A_p,\cdots,&\tilde A_{-p}, \tilde B_p,\cdots,\tilde B_{-p}]\\&-[ A_p,\cdots, A_{-p}, B_p,\cdots,B_{-p}]\|\leq 4\epsilon M\end{align*}
where $M=(n+m)(2p+1)\|V^{-1}\|$. 
\end{corollary}
\begin{proof}
 As the rank condition \eqref{rank} is achieved, an invertible matrix $V$ can be extracted from the columns of $(\mathbf{X_0},\mathbf{U_0})$.
 As for any $Y\in \mathbb{R}^r$ with $r=(n+m)(2p+1)$ s.t. $\|Y\|=1$, there exists $\Lambda$ s.t. $Y=V\Lambda$, the following relation holds (using \eqref{error}): 
 \begin{align*}
 \|([\tilde A_p,&\cdots,\tilde A_{-p}, \tilde B_p,\cdots,\tilde B_{-p}]\\&-[ A_p,\cdots, A_{-p}, B_p,\cdots,B_{-p}])Y\|\leq 4\epsilon \sum_i |\Lambda_i|
 \end{align*}
Using norm equivalence in finite dimension, we have:
\begin{align*}
 \|([\tilde A_p,&\cdots,\tilde A_{-p}, \tilde B_p,\cdots,\tilde B_{-p}]\\&-[ A_p,\cdots, A_{-p}, B_p,\cdots,B_{-p}])Y\|\leq 4\epsilon r \|\Lambda\|\\
 &\leq 4\epsilon r \|V^{-1}Y\|\leq 4\epsilon r \|V^{-1}\|
 \end{align*}
Taking the supremum on $Y$ leads to the result. 
\end{proof}

\begin{remark} Recall that Corollary~\ref{conver2} primarily establishes the existence of a solution, yet it does not provide a constructive approach. Determining an appropriate value for $p$ often entails a trial-and-error process until a satisfactory solution is achieved. Here, a satisfactory solution denotes one in which the higher-order phasors obtained approach zero.
Furthermore, the rank condition \eqref{rank} acts as a necessary condition for the uniqueness of a solution in Problem~\ref{lsqr} when dealing with noise-free data. However, it is essential to highlight that this rank condition becomes insufficient in the presence of noisy measurements. To mitigate the impact of noise, a larger value for $L$ becomes imperative in addressing Problem~\ref{lsqr}.
\end{remark}

\subsection{Validation}
Before exploring specific examples, we introduce a validation protocol that will be applied in the following section. For theoretical validation, our initial focus is on the noiseless scenario. After successfully solving the least-squares Problem~\ref{lsqr} for a particular value of $p$, the subsequent step is to validate the resulting model. Here, we employ the relative error between the true and estimated phasors as our validation criterion. This relative error is computed as the following percentage:
 
\begin{equation}
 err := 100 \cdot \frac{||P_{th}-P_{est}||_{2}}{||P_{th}||_{2}} \label{err}
\end{equation}
where $P_{th}$ (respectively $P_{est}$) is a matrix which contains the theoretical values of the phasors (respectively the phasors estimated with the least-squares method). %It can be assumed that $||P_{th}||_{2} \neq 0$, or else the system would have no dynamics worth studying.
A threshold $\varepsilon$ must be chosen to indicate if the obtained values of the phasors are close enough to the true ones. This threshold depends on the intended use of the identified model. If $err \leq \varepsilon$, the estimated model is acceptable, otherwise a larger value of $p$ or $L$ must be chosen to satisfy the validation criterion.

It is crucial to recognize that in real-world scenarios, the theoretical values of the phasors are often unknown. Consequently, an alternative validation criterion must be employed. This validation procedure involves the use of a distinct dataset separate from the one used for identification. One can simulate the LTP system on a new trajectory, facilitating a comparison between the true and estimated trajectories.

\vspace{-.2cm}
\section{Illustrative examples}

In this section, we illustrate our approach on three examples, one of which is the wind turbine problem borrowed from \cite{bottasso_model-independent_2015}. To assess the sensitivity of our proposed methodology to noisy data, we first test the protocol with noiseless data, then we introduce a random disturbance to the state measurements $x$. At each time instant $t$ and for $i=1,...,n$, the noise on state $x_i$ conforms to a Gaussian distribution $N(0,\sigma^2)$ where $3 \sigma = \frac{5}{100} |x_i (t)|$, so that a certain signal-to-noise ratio is achieved. For each example, we solve the least-squares Problem~\ref{lsqr} with a set of $100$ random initial conditions $x(t_0)$. If the system is non autonomous, piecewise-periodic signals are chosen as inputs, with phasors acquired from a normal distribution.

\subsection{A finite phasor-order example}

First, let us consider a LTP system generated with random phasors for $A$ and $B$, where $n=3$, $m=2$ and $d^oA = d^oB =10$. Consequently, the truncation order can be set to $p = 10$. This leads to a total of $n(n+m)(2p+1) = 315$ unknowns. 

In the absence of noise, the choice of $L = (n+m)(2p+1)$ and $\delta t = \frac{T}{4p}$ enables precise recovery of the unknown phasor values, provided that the system is sufficiently excited by the input, as indicated by the fulfillment of condition \eqref{rank}. In such instances, the relative error across all $100$ trials remains below $10^{-6} \%$. This outcome underscores the success of the identification protocol when well-selected inputs are employed.

In the presence of noise, errors arise during the computation of $\dot{X}_{0}(t)$ and $X_{k}(t)$. Since $\dot{X}_{0}(t)$ is determined using formula \eqref{dotX0}, and the state noise is zero-mean, the error associated with $\dot{X}_{0}(t)$ also exhibits a zero-mean characteristic. The Fast Fourier transform introduces a noise-smoothing effect during the calculation of $X_{k}(t)$. Moreover, the disturbance is bounded. Consequently, if we denote the data matrices affected by noise as $\tilde{\mathbf{X}}_1$ and $\tilde{\mathbf{X}}_0$, there exist values, $\varepsilon_1$ and $\varepsilon_0$, such that $|| \mathbf{X_1} - \tilde{\mathbf{X}}_1 || \leq \varepsilon_1$ and $|| \mathbf{X_0} - \tilde{\mathbf{X}}_0 || \leq \varepsilon_0$.

A larger value of $L$, namely $3(n+m)(2p+1)$, is chosen to ensure that a precise enough solution can be found. The relative error defined in equation \eqref{err}, calculated for the identified phasors, falls within the range of $3.2 \%$ and $8.5 \%$ across all $100$ trials. 

The validation of the identification results on a new trajectory is depicted in Fig. ~\ref{ex}. With noise free data, the true and estimated trajectories align closely. In the presence of noise, an approximation of the true trajectories remains possible.

\begin{figure}[h!]\begin{center}
		\includegraphics[width=\linewidth]{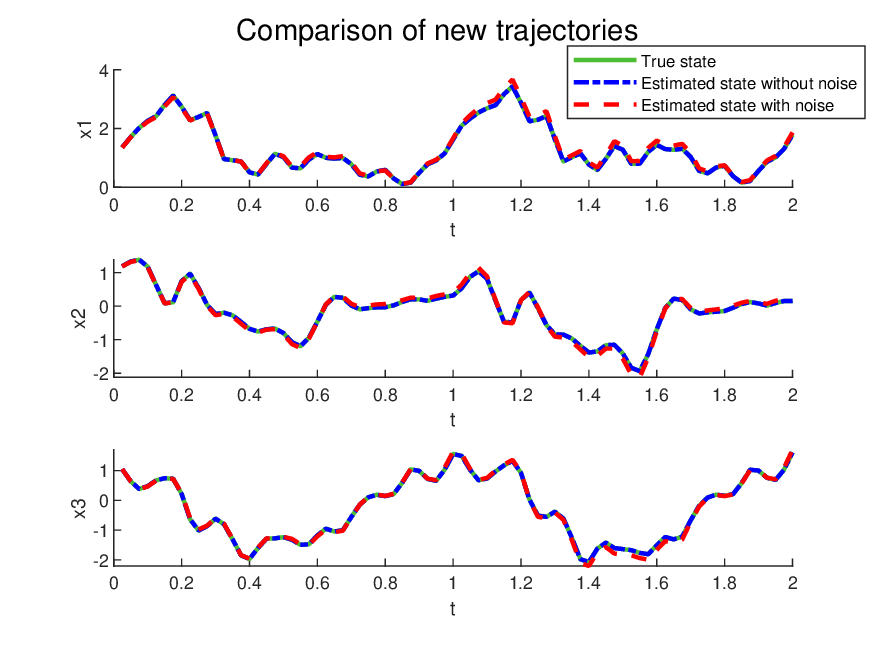}
		\caption{Comparison of the true and estimated trajectories of the finite phasor-order example.} \label{ex}
	\end{center}
\end{figure}

\subsection{An infinite phasor-order example}
Consider the example discussed in \cite{riedinger2022solving}:
\begin{align}
	\dot x=&\left(\begin{array}{cc}a_{11} (t) & a_{12} (t) \\a_{21} (t) & a_{22} (t)\end{array}\right)x+\left(\begin{array}{c}b_{11}(t) \\0\end{array}\right)u\label{ex_inf}\end{align}
{\small\begin{align*}a_{11} (t) &=1+\frac{4}{\pi}\sum_{k=0}^{\infty}\frac{1}{2k+1}\sin(\omega (2k+1)t),\\
	a_{12} (t) &= 2+\frac{16}{\pi^2}\sum_{k=0}^{\infty}\frac{1}{(2k+1)^2}\cos(\omega (2k+1)t),\\
	a_{21} (t) &= -1+\frac{2}{\pi}\sum_{k=1}^{\infty}\frac{(-1)^k}{k}\sin(\omega kt+\frac{\pi}{4}),\\
	a_{22} (t) &= 1-2\sin(\omega t)-2\sin(3\omega t)+2\cos(3\omega t)+2\cos(5\omega t),\\
	b_{11}(t)&=1+ 2 \cos(2\omega t)+ 4 \sin(3\omega t) \text{ with }\omega=2\pi.
\end{align*}}
The corresponding Toeplitz matrix $\mathcal{A}$ possesses an infinite number of phasors and does not exhibit a banded structure. Nevertheless, its higher-order phasors converge towards zero, which implies that it remains feasible to identify the non-negligible phasors within the matrix.

Given the inherent instability of this system, we employ data obtained from multiple trajectories for phasor identification (see \cite{van2020willems}). With a truncation order set at $p=25$, a time step of $\delta t = \frac{T}{256}$ and utilising $16$ trajectories, each with a length of $512$ time points, the identification error ranges from $4.6 \%$ to $9.8 \%$ among the 100 trials. 

The validation of the identification results on a new trajectory is visually demonstrated in Fig. ~\ref{ex2}. It is noteworthy that the algorithm can operate with a reduced number of trajectories if their length are chosen according to \cite{van2020willems}.

\begin{figure}[h!]\begin{center}
		\includegraphics[width=\linewidth,height=6cm]{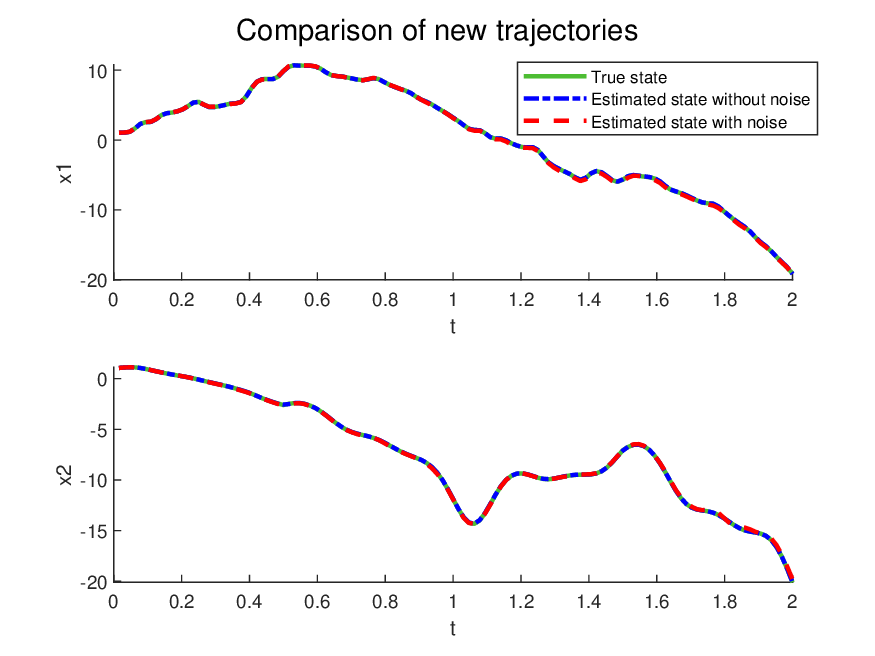}
		\caption{Comparison of the true and estimated trajectories of the infinite phasor-order example.} \label{ex2}
	\end{center}
\end{figure}

\begin{figure}[h!]\begin{center}
		\includegraphics[width=\linewidth]{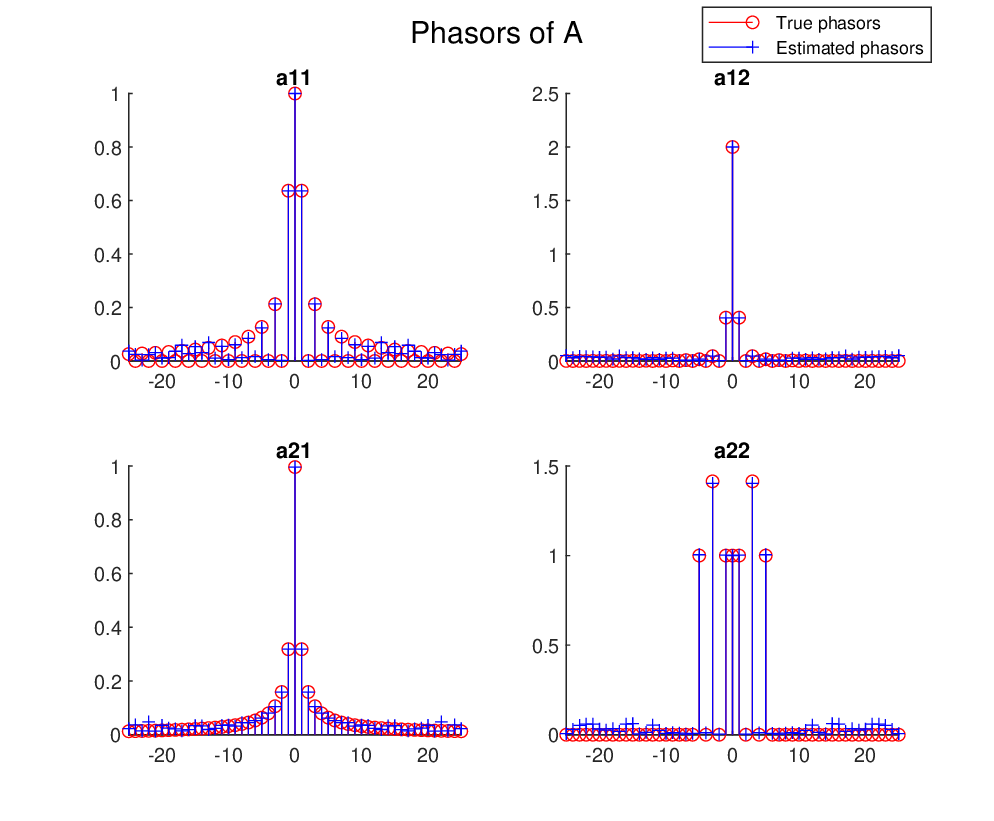}
		\caption{Moduli of the true and estimated phasors of $A$ with noise on the state measurements.} \label{ex2bis}
	\end{center}
\end{figure}

Let us plot the true and estimated moduli of the phasors for $A$ on Fig. ~\ref{ex2bis}. Since we tackled the identification problem with a finite truncation order $p=25$ and there was noise in the identification data, the effects of this noise are discernible, particularly on the higher-order phasors. This phenomenon is to be expected, as the chosen noise injects high-frequency content into the signal. However, the amplitude of these components is obviously limited. Figure ~\ref{ex2} shows that the predictive accuracy of the estimated phasors is already acceptable. To improve the results, a threshold can be set on the phasor modulus so that phasors below this threshold are eliminated.
In summary, even with a theoretically infinite number of phasors, it is possible to obtain an approximation for the non-negligible phasors.

\subsection{Wind turbine}
Consider the three-bladed wind turbine discussed in \cite{bottasso_model-independent_2015}, and review its key attributes. The equation governing the system's motion can be expressed as follows:
 \begin{equation*}
 M(t) \Ddot{q}(t) + C(t) \dot{q}(t) + K(t) q(t) = 0 
\end{equation*} where $M$, $C$ and $K$ are the system's mass, damping and stiffness matrices, and $q(t) = (\zeta_1 (t),\zeta_2 (t), \zeta_3 (t), y_{c}(t))^{T}$ contains the lag angles of each blade and the horizontal displacement of the hub. The rated rotor speed is $\Omega_r =1.2 \ rad.s^{-1}$. Functions $M$, $C$ and $K$ are periodic with period $\frac{2 \pi}{\Omega_r}$.
This autonomous system can be described by the state equation:
$
 \dot{x}(t) = A(t) x(t)
$ where 
\begin{equation*}
 x(t) = \begin{pmatrix}
 q(t) \\ \dot{q}(t)
\end{pmatrix},\ A(t) = \begin{pmatrix} 0_{4,4} & I_{4} \\ -M(t)^{-1} K(t) & -M(t)^{-1} C(t) \end{pmatrix}
\end{equation*}
We are dealing with an unstable system characterized by $n=8$. In this context, we assume a truncation order of $p=4$, resulting in the determination of $n^2 (2p+1) = 576$ unknowns. The step size is set to $\delta t = \frac{T}{256}$.

Given the inherent instability of this system, we rely on data from $15$ trajectories, each with a length of $256$ to identify its phasors. 
The identification protocol proves successful across all initial conditions, with the relative error \eqref{err} ranging from a minimum of $0.9 \%$ to a maximum of $6.3 \%$ among the $100$ trials. Furthermore, the validation of these identification results on a new trajectory is depicted in Fig.~\ref{wind_turbine}.

\begin{figure}[h!]\begin{center}
		\includegraphics[width=\linewidth]{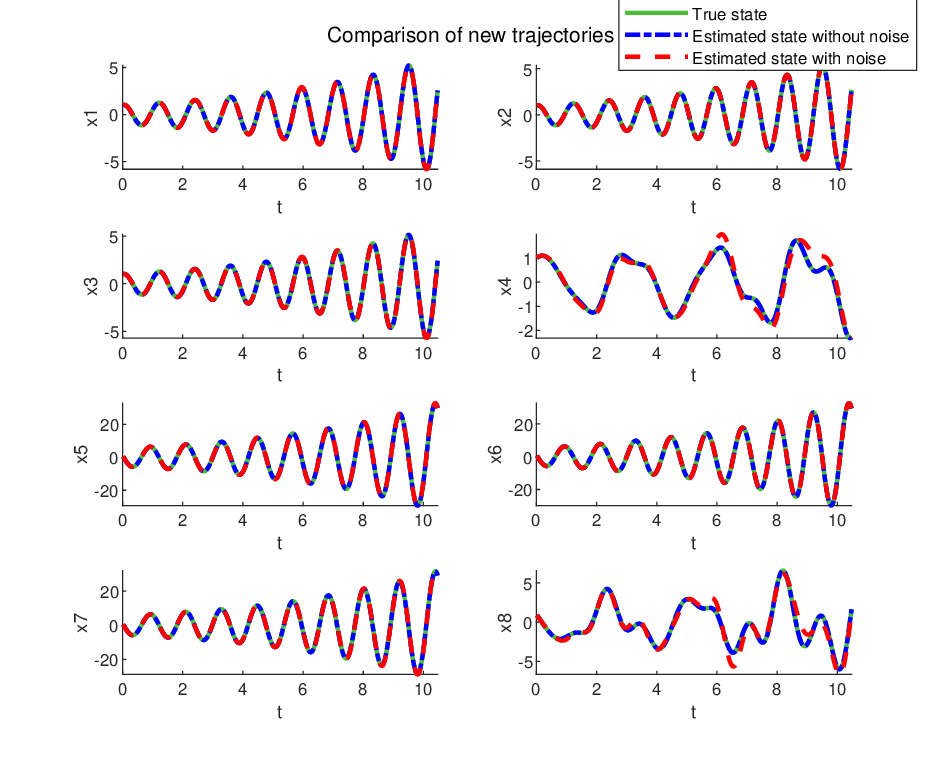}
		\caption{Comparison of the true and estimated trajectories of the wind turbine.} \label{wind_turbine}
	\end{center}
\end{figure}

The curve of the mean identification error is represented on Fig.~\ref{error_conv}, which shows that the precision of the solution is improved by a larger value of $L$.

\begin{figure}[h!]\begin{center}
		\includegraphics[width=\linewidth]{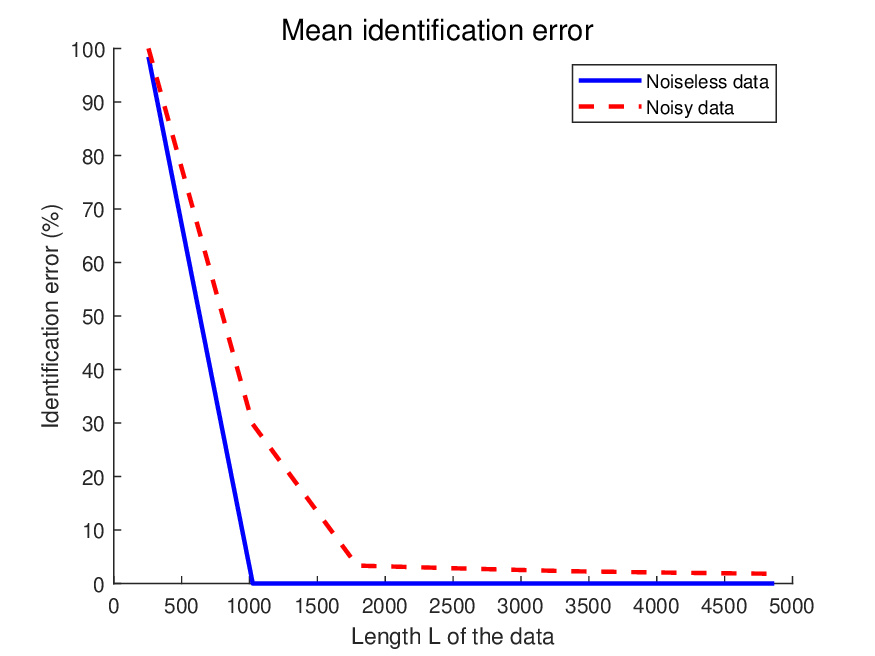}
		\caption{Identification error convergence for the wind turbine.} \label{error_conv}
	\end{center}
\end{figure}

Given that the system under investigation is a wind turbine, achieving precise identification holds significant relevance, particularly in the context of control applications.
%\addtolength{\textheight}{0cm} % This command serves to balance the column lengths
\section{Conclusion}
In this paper, we have presented a novel approach that enables the identification of the state and input matrices of a LTP system up to an arbitrarily small error. Our approach capitalizes on the intrinsic equivalence between LTP systems in the time domain and infinite-dimensional LTI systems in the harmonic domain. Leveraging the block Toeplitz structure of the latter, we have devised a finite-dimensional linear least-squares problem, the solution of which corresponds to the unknown phasors. 
Our approach offers significant advantages, including avoiding signal derivative calculations, and performs effectively in noisy scenarios even for unstable systems, as shown in numerical simulations.

\bibliographystyle{ieeetr}
\bibliography{references}

\end{document}